\documentclass[conference]{IEEEtran}
\makeatletter
\def\ps@headings{%
\def\@oddhead{\mbox{}\scriptsize\rightmark \hfil \thepage}%
\def\@evenhead{\scriptsize\thepage \hfil \leftmark\mbox{}}%
\def\@oddfoot{}%
\def\@evenfoot{}}
\makeatother
\usepackage{amsfonts}
\usepackage{times}
\usepackage{latexsym}
\usepackage{amssymb}
\usepackage{amsmath}
\usepackage{cite}
\usepackage{verbatim}

\newcommand{\bydef}{\triangleq}

\def\bydef{:=}

\def\bb0{{\mathbb{0}}}

\def\bydef{:=}

\def\bb{{\mathbf{b}}}

\def\bh{{\mathbf{h}}}

\def\bs{{\mathbf{s}}}

\def\b0{{\mathbf{0}}}

\def\bB{{\mathbf{B}}}

\def\bH{{\mathbf{H}}}

\def\bS{{\mathbf{S}}}

\def\bbE{{\mathbb{E}}}

\def\bbN{{\mathbb{N}}}

\def\bbR{{\mathbb{R}}}

\def\bydef{:=}

\def\sf0{{\mathsf{0}}}

\def\nn{\nonumber}

\usepackage{graphicx}
\usepackage{amssymb}
\usepackage{amsfonts}
\usepackage{amsmath}
\usepackage{latexsym}
\usepackage{epstopdf}

\begin{document}

\newtheorem{thm}{Theorem}
\newtheorem{lemma}{Lemma}
\newtheorem{rem}{Remark}
\newtheorem{exm}{Example}
\newtheorem{prop}{Proposition}
\newtheorem{defn}{Definition}
\newtheorem{cor}{Corollary}
\def\proof{\noindent\hspace{0em}{\itshape Proof: }}
\def\endproof{\hspace*{\fill}~\QED\par\endtrivlist\unskip}
\def\bh{{\mathbf{h}}}
\def\SIR{{\mathsf{SIR}}}
\def\SINR{{\mathsf{SINR}}}

\title{Percolation and Connectivity on the Signal to Interference Ratio Graph}
\author{
Rahul~Vaze\\ School of Technology and Computer Science,\\ Tata Institute of Fundamental Research, \\ Homi Bhabha Road, Mumbai 400005, \\vaze@tcs.tifr.res.in. }

\date{}
\maketitle
\noindent
\begin{abstract}
A wireless communication network is considered where any two nodes are connected if the  signal-to-interference ratio (SIR) between them is 
greater than a threshold. Assuming that the nodes of the wireless network are distributed as a Poisson point process (PPP), 
percolation (formation of an unbounded connected cluster) on the resulting SIR graph is studied as a function of the density of the PPP. It is shown that for a small enough threshold, there exists a closed interval of densities for which percolation happens with non-zero probability. Conversely, it is  shown that for a large enough threshold, there exists a closed interval of densities for which the probability of percolation is zero. 
Connectivity properties of the SIR graph are also studied by restricting all the nodes to lie in a bounded area. 
Assigning separate frequency bands or time-slots proportional to the logarithm of the number of nodes to different nodes for transmission/reception is shown to be necessary and sufficient for guaranteeing connectivity in the SIR graph.
\end{abstract}

\section{Introduction} 
Consider  a large ad-hoc wireless network, where multiple transmitter receiver pairs communicate simultaneously in an
uncoordinated manner without the help of any fixed infrastructure. Important examples of ad-hoc networks
include vehicular networks, military and emergency networks, and sensor networks. The uncoordinated
nature of communication allows multiple transmitters to communicate at the same time, however,  
creates interference at all receivers. A common connection model in an ad-hoc network is the signal-to-interference ratio (SIR) model,\footnote{Ignoring the additive noise in an interference limited system.} where two nodes are connected if the SIR between them is greater than a threshold. In this paper we are interested in studying 
the probability of the formation of an unbounded connected cluster with the SIR model in an ad-hoc network. The study is motivated by the fact that the 
presence of unbounded connected clusters guarantees long range connectivity using multi-hop routing in an ad-hoc wireless network.

A natural tool to study the formation of unbounded connected clusters in a graph associated with a wireless network is  percolation theory 
 \cite{BookRoy}, where  percolation is defined as the event that there exists an unbounded connected cluster in a graph. 
Previously, assuming the location of nodes of the wireless network to be distributed as a Poisson point process (PPP) with density $\lambda$, percolation has been studied for the Boolean model \cite{Gilbert1961}, where two nodes are connected if the two circles drawn around them with a fixed radius overlap, for the random Boolean model \cite{BookRoy, Gourre2008}, where two nodes are connected if the two circles drawn around them with a random radius overlap, for the random connection model \cite{Penrose1991}, where two nodes are connected  with some probability which depends on the distance between them independently of other nodes. 
For all these connection models, a phase transition behavior has been established in \cite{BookRoy, BookPenrose, Gourre2008, Gilbert1961,Penrose1991}, i.e. there exists a critical density $\lambda_c$, where if $\lambda < \lambda_c$, then the probability of percolation is zero, while if $\lambda > \lambda_c$ then percolation happens almost surely. In other words, percolation is shown to be monotonic in $\lambda$ \cite{BookRoy, Gourre2008, Gilbert1961,Penrose1991}.

The most relevant work to the present paper is \cite{Dousse2006}, (an improved version of \cite{Dousse2005}) where percolation on the SINR graph (constructed from an underlying wireless network with nodes distributed as a PPP) has been studied. In \cite{Dousse2006}, the SINR graph is defined to be $\{\Phi, {\cal E}\}$, where $\Phi$ is the set of nodes, and the edge set ${\cal E} = \{(x_i,x_j) \ : \ \SINR_{ij} > T\}$, with $\SINR_{ij} \bydef  \frac{ g(d_{ij}) } 
{\sigma^2 + \sum_{k\in \Phi, k\ne i }\gamma g(d_{kj})}$, where $d_{kj}$ is the distance between nodes $x_k$ and $x_j$, $g(.)$ is the signal attenuation 
function, 
$\sigma^2$ is the variance of the AWGN, $T$ is the connection threshold, and $\gamma > 0$ is an interference suppression parameter that depends on the wireless technology e.g. CDMA.  In \cite{Dousse2006}, it is shown that if  $\lambda_c$ is the critical density with $\gamma=0$, then there exists a $\gamma* >0$, such that for any $\lambda > \lambda_c$, percolation happens in the SINR graph for $\gamma < \gamma^*$. Lower and upper bounds on 
$\gamma^*$ have been obtained in \cite{Yeh2007}.
Thus, \cite{Dousse2006} 
shows that there exists a small enough $\gamma$ for which the percolation properties of the SINR graph are similar to $\gamma=0$. Note that with respect to $\gamma$, SIR graph percolation is monotonic, since if percolation happens for $\gamma_0$, then percolation happens for all $\gamma < \gamma_0$. Even though \cite{Dousse2006}
provides key insights into the percolation properties of the SINR graph, however, its scope is limited since assuming arbitrarily small enough $\gamma$ is not feasible from any wireless technology perspective. 

In this paper we consider $\gamma =1$,  and ignore the additive noise contribution, since with $\gamma=1$, the system is interference limited. Ignoring the noise contribution, with $\gamma =1$, $\SINR_{ij} = \SIR_{ij} \bydef  \frac{ g(d_{ij}) } 
{ \sum_{k\in \Phi, k\ne i } g(d_{kj})}$. Assuming that the nodes of $\Phi$ are distributed as a homogenous PPP, in this paper we are interested in finding the range of $\lambda$'s for which percolation happens in the SIR graph.  


As discussed before, typically, continuum percolation exhibits phase transition behavior and is monotonic in the quantity of interest, e.g. monotonic in  $\lambda$ \cite{Gilbert1961, Penrose1991}, monotonic in  $\gamma$ \cite{Dousse2006}.
The continuum percolation on the SIR graph, however, does not seem to be monotonic in $\lambda$.
To illustrate this, let  percolation happen for some value of $\lambda$, say $\lambda_0$.  
Then increasing $\lambda$ beyond $\lambda_0$, the distance between the nodes decreases, and hence both the signal and the interference powers
increase simultaneously. Thus, it is difficult to establish that percolation happens for any $\lambda>\lambda_0$ for a fixed $T$. The only cases where it is trivial to establish whether percolation happens or not are: $\lambda=0$, or $T=\infty$, (no percolation) and $T=0$ (percolation). 
Moreover, it is also not obvious whether percolation happens for any value of $\lambda$ for a fixed $T$.

In this paper  for the path-loss model, we show that for large enough $T$, there exists a closed interval $\Lambda_T^{l} \bydef [\lambda_T^{l1} \ \lambda_T^{l2}]$, such that  if $\lambda \in \Lambda_T^{l}$, then the probability of percolation is zero (sub-critical regime). In \cite{Dousse2005}, where a link between $x_i$ and $x_j$ is defined in the SINR graph if both $\SINR_{ij}$ and $\SINR_{ji}$ are greater than the same threshold $T$, it is shown that if $\gamma > \frac{1}{T}$, then the probability of percolation is zero.  In this paper, we consider that a link exists between 
$x_i$ and $x_j$ in the SIR graph if $\SIR_{ij} >T$, which is a relaxed condition compared to \cite{Dousse2005}, and 
consequently the analysis and results of  \cite{Dousse2005} cannot be used to derive bounds on the sub-critical regime.

Conversely, we show that for small enough $T$, there exists a closed interval $\Lambda_T^{u}\bydef [\lambda_T^{u1} \ \lambda_T^{u2}]$, such that  if  $\lambda \in \Lambda_T^{u}$, then the percolation happens with non-zero probability. Our result loosely establishes continuity of percolation at $T=0$, since at $T=0$ percolation happens for all non-zero values of $\lambda$. One might argue that a small enough $T$ is also not practical, since the rate of transmission 
between any pair of nodes depends on $T$.  
Our result essentially establishes that for some node intensities, an infinite connected component can be formed in a wireless network, where 
each link has a small rate of transmission. For example, in a delay tolerant network, where reliability is more important than the rate of information transfer, our results show that large data transfers can be made  to a large enough number of nodes by using low rate links with strong error correcting codes.

Even though percolation guarantees the formation of unbounded clusters, it does not ensure connectivity between any two nodes of the network. In a wireless network, connectivity is quite critical, and studying connectivity properties of large networks (formally defined to be event that there is a path between any pair of nodes) has received a lot of attention in literature \cite{Gupta1998, Avin2010, Gupta2000,Blaszczyszyn2010, Penrose1999, DousseBaccelli05, BookRoy}, primarily for the Boolean model of connectivity. For studying connectivity in the SIR graph, we restrict ourselves to a finite area, to be precise an unit square, since the probability of connectivity when nodes are distributed on an infinite plane is zero. We assume that there are $n$ nodes lying in the unit square that are independently drawn from an uniform distribution over the unit square. We consider the case when $C(n)$ separate frequency bands/time slots (called colors) are used by the $n$ nodes for transmission and reception, where only signals belonging to the same color 
interfere with each other. 
We show that $C(n) = \kappa \log n$ ($\kappa$ is a constant) is necessary and sufficient for ensuring the connectivity of the SIR graph with high probability. The result suggests that if there are order $\frac{n}{\log n}$ interferers for any receiving node, then the SIR between a large number of node pairs can be guaranteed to be above a constant threshold. 


{\it Notation:}
The expectation of function $f(x)$ with respect to $x$ is denoted by
${\bbE}(f(x))$.
A circularly symmetric complex Gaussian random
variable $x$ with zero mean and variance $\sigma^2$ is denoted as $x
\sim {\cal CN}(0,\sigma^2)$.  $(x)^+$ denotes the function $\max\{x,0\}$. $|S|$ denotes the cardinality of set $S$. The complement of set $S$ is denoted by $S^c$. $S_2\backslash S_1$ represents the elements of $S_2$ that are not in its subset $S_1$. We denote the 
origin by ${\bf 0}$. 
A ball of radius $r$ centered at $x$ is denoted by $\bB(x,r)$.  The set $\{1,2,\ldots, N\}$ is denoted by $[N]$.
 We use the symbol
$\bydef$  to define a variable. We define 
$f(n) = {\cal O}(g(n))$ if $\exists \ k > 0, \ n_0, \ \forall \ n>n_0$, $|f(n)| \le |g(n)| k$.

\section{System Model}
\label{sec:sys}
Consider a wireless network with the set of nodes denoted by $\Phi$. 
For $x_i, x_j \in \Phi$, let $d_{ij}$ denote the distance between $x_i$ and $x_j$. We assume that if power $P$ is transmitted by node $x_i$, then the received signal power at $x_j$ is $Pg(d_{ij})$, where $g(.)$ is the monotonically decreasing signal attenuation function with distance.\footnote {The most commonly found signal attenuation function in literature is $g(x) = x^{-\alpha}$, however, it is singular at distances close to zero, and results in $\int_0^{\infty} xg(x)  dx = \infty$. However, owing to simplicity of exposition, we use $g(x) = x^{-\alpha}$, except for Subsection \ref{sec:plsupc} and Subsection \ref{sec:lbconn}, where any monotonically decreasing $g(.)$  with 
$\int_0^{\infty} xg(x)  dx < \infty$ is considered.} 
With concurrent transmissions from all nodes of $\Phi$, the received signal at $x_j$ at any time is  
\begin{equation}
\label{rxsig}
r_j = \sum_{k\in \Phi, k\ne i }\sqrt{Pg(d_{kj})}s_k + v_{j}, 
\end{equation}
where $s_k$ is signal transmitted from node $x_k$, $P$ is the power transmitted by each node, and $v_{j}$ is the AWGN with ${\cal CN}(0,1)$ distribution. Note that this is an interference limited system, and we drop the contribution of the AWGN in the sequel. From (\ref{rxsig}), the SIR for the $x_i$ to $x_j$ 
communication is $\SIR_{ij} \bydef  \frac{g(d_{ij})} 
{\sum_{k\in \Phi, k\ne i }g(d_{kj})}$.
We consider the SIR graph  of \cite{Dousse2006}, where an edge between $x_i$ and $x_j$,  $x_i,x_j\in\Phi$, exists if 
the SIR between $x_i$ and $x_j$, $\SIR_{ij}$, is greater than a threshold  $T$.


\begin{defn} SIR graph is a directed graph $SG(T) \bydef \{\Phi, {\cal E}\}$, with vertex set $\Phi$, and edge 
set ${\cal E} \bydef \{(x_i, x_j): \SIR_{ij}\ge T\}$, where $T$ is the SIR threshold required for correct decoding required between any two nodes of $\Phi$. 
\end{defn}

\begin{defn} We define that there is a {\it path} from node $x_i$ to $x_j$ if there is a path from $x_i$ to $x_j$ in the $SG(T)$.  A path between $x_i$ and $x_j$ on $SG(T)$ is represented as $x_i \rightarrow x_j$.
\end{defn}

\begin{defn} We define that a node $x_i$ can {\it connect} to  $x_j$ if there is an edge between $x_i$ and $x_j$ in the $SG(T)$.\end{defn}


Similar to \cite{Dousse2006}, in this paper we assume that the locations of $\Phi$  are distributed as a homogenous 
Poisson point process (PPP) with 
density $\lambda$. 
The SIR graph  when $\Phi$ is  distributed as a PPP is referred to as the Poisson SIR graph (PSG).
We define the connected component of any node $x_j \in \Phi$, as $C_{x_j} \bydef   \{x_k \in \Phi, x_j\rightarrow x_k\}$, with cardinality  $|C_{x_j}|$. Note that because of stationarity of the PPP, the distribution of $|C_{x_j}|$ does not depend on $j$, and hence without loss of generality from here on we consider node $x_1$ for the purposes of defining connected components.

In this paper we are interested in studying the percolation properties of the PSG. In particular, 
we are interested in finding the values of $\lambda$ for which the probability of having an unbounded connected component in PSG is greater than zero, i.e. $\lambda_p \bydef  \{{\lambda}:P(|{\cal C}_{x_1}| = \infty)>0\}$. 
The event  $\{|{\cal C}_{x_1}| = \infty\}$ is also referred to as {\it percolation} on PSG, and we say that percolation happens if 
$P(\{|{\cal C}_{x_1}| = \infty\})>0$, and does not happen if $P(\{|{\cal C}_{x_1}| = \infty\})=0$. Ideally, we would like to find sharp cutoff $\lambda_c$ for $\lambda$ as a function of $T$, such that for $\lambda > \lambda_c$ $P(\{|{\cal C}_{x_1}| = \infty\})>0$, while with $\lambda \le \lambda_c$ $P(\{|{\cal C}_{x_1}| = \infty\})=0$. This problem, however, is quite challenging, and in this paper we only establish that for large enough $T$ there exists a closed interval $\Lambda_T^{l} \bydef [\lambda_T^{l1} \ \lambda_T^{l2}]$, such that  if $\lambda \in \Lambda_T^{l}$ then the probability of percolation is zero, while for small enough $T$ there exists a closed interval $\Lambda_T^{u}\bydef [\lambda_T^{u1} \ \lambda_T^{u2}]$, such that  if  $\lambda \in \Lambda_T^{u}$ 
then the probability of percolation is greater than zero.

\begin{rem} Note that we have defined PSG to be  a directed graph, and the component of $x_1$ is its out-component, i.e. the set of nodes with which $x_1$ can communicate. 
Since $x_i \rightarrow x_j, \ x_i,x_j\in \Phi$, does not imply $x_j \rightarrow x_i \ x_i,x_j\in \Phi$,  one can similarly define in-component 
$C_{x_j}^{in} \bydef   \{x_k \in \Phi, x_k\rightarrow x_j\}$,  bi-directional component $C_{x_j}^{bd} \bydef   \{x_k \in \Phi, x_k\rightarrow x_j\ \text{and} \ x_k\rightarrow x_j\}$,  and either one-directional component $C_{x_j}^{ed} \bydef   \{x_k \in \Phi, x_k\rightarrow x_j \ \text{or} \ x_k\rightarrow x_j \}$.  
\end{rem}

 
  \section{Percolation on the SIR graph} 
  \label{sec:pl}
In this section, we first discuss the sub-critical regime where the probability of percolation is zero, and then follow it up the super-critical regime where the probability of percolation is greater than zero.
  
\subsection{Sub-critical regime}
\label{sec:plsubc} For simplicity of exposition, in this subsection we assume that the signal attenuation function $g(d_{ij}) = d_{ij}^{-\alpha}$, where $\alpha >2$ is the path-loss exponent. The results of this subsection can be extended to  any 
signal attenuation function $g(.)$ that is monotonically decreasing and has $\int xg(x)dx < \infty$. Thus, in this case, $\SIR_{ij} =  \frac{d_{ij}^{-\alpha}}{\sum_{k\in \Phi, k\ne i }d_{kj}^{-\alpha}}$. Let $I_j^i \bydef \sum_{k\in \Phi, k\ne i}d_{kj}^{-\alpha}$, then  $PSG = \{\Phi, {\cal E}\}$, where the edge set ${\cal E} = \left\{(x_i, x_j): d_{ij}\le \left(\frac{1}{T I^i_j}\right)^{\frac{1}{\alpha}}\right\}$.  In this subsection we are interested in deriving that for large enough $T$ there exists a closed interval $\Lambda_T^{l} \bydef [\lambda_T^{l1} \ \lambda_T^{l2}]$, such that  if $\lambda \in \Lambda_T^{l}$ then the probability of percolation is zero.

\begin{figure}
\centering
\includegraphics[width=3.5in]{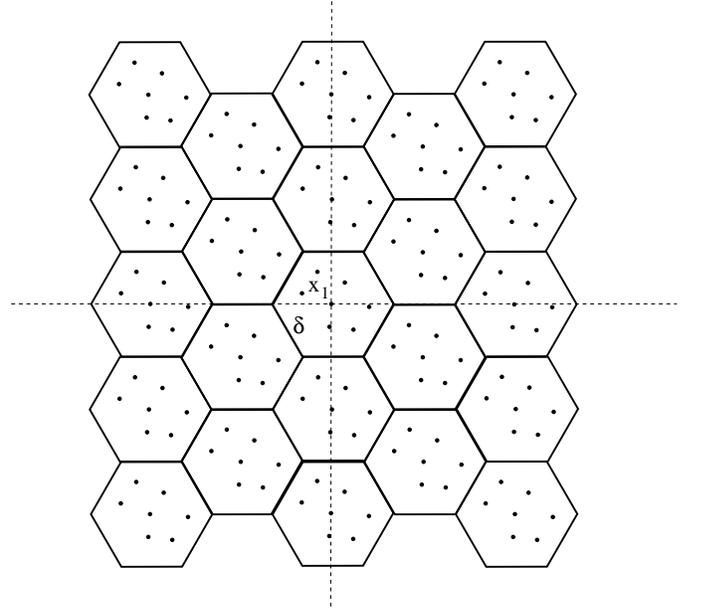}
\caption{Two dimensional hexagonal lattice with edge $\delta$.}
\label{fig:lattice}
\end{figure}

\begin{figure}
\centering
\includegraphics[width=3in]{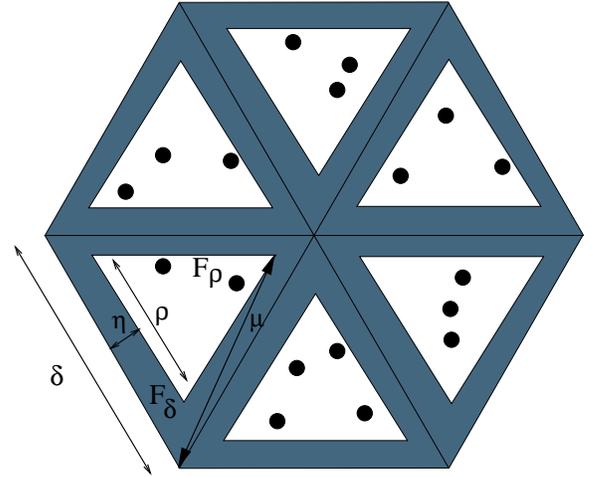}
\caption{Closed face of the hexagonal lattice.}
\label{fig:closedface}
\end{figure}

\begin{figure}
\centering
\includegraphics[width=3.5in]{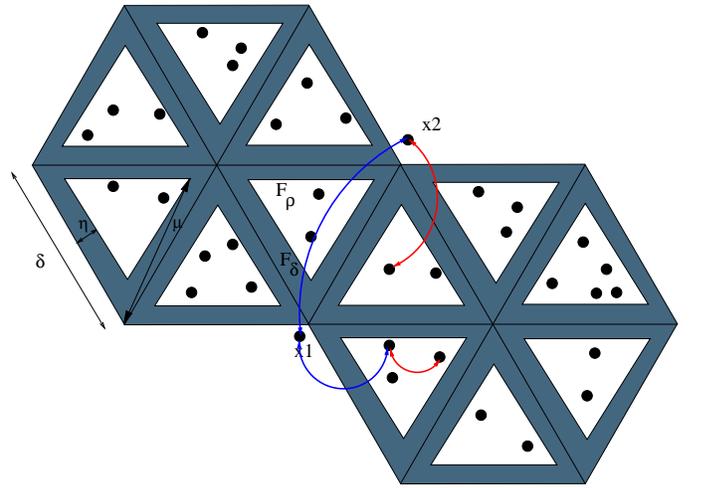}
\caption{Node disconnection  because of closed face.}
\label{fig:closedfaceeffect}
\end{figure}

\begin{figure}
\centering
\includegraphics[width=3.5in]{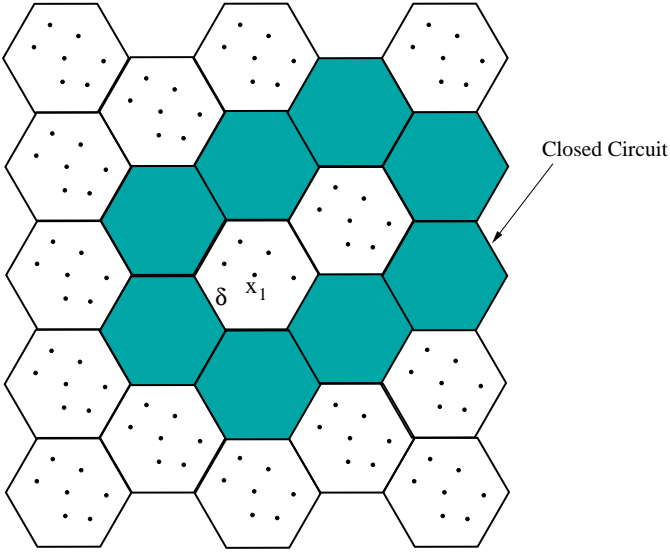}
\caption{Connected node partitioning because of a closed circuit.}
\label{fig:closedcircuit}
\end{figure}

Towards that end, we tile $\bbR^2$ 
using an hexagonal lattice $\bH$ with edge $\delta$ as shown in Fig. \ref{fig:lattice}. We let node $x_1$ to lie on the origin of $\bH$. 
Each face of the hexagonal lattice has two states, either {\it open} or {\it closed}. As shown in Fig. \ref{fig:closedface},
a face of $\bH$ is defined to be closed if each of the six equilateral triangles $F_{\delta}$ inside each face are closed, and  $F_{\delta}$ is defined to be  closed if  
\begin{enumerate}
\item there is no node of $\Phi$ in the shaded region $F_{\delta} \backslash F_{\rho}$.
\item there are at least two nodes of $\Phi$ in the inner equilateral triangle $F_{\rho}$.
\item $\rho \le \eta  T^{\frac{1}{\alpha}}$.
\item $\mu \le \delta T^{\frac{1}{\alpha}}$.
\end{enumerate} 

With these definitions we can map the continuum percolation on the PSG to discrete percolation on the hexagonal lattice. 
Conditions 1) and 4) together imply that no two nodes on either side of a closed face of $\bH$ can have an edge between each other. To see this, let $x_i$ lie on the left side of any closed face of $\bH$ and $x_j$ lie on the right of the closed face. See Fig. \ref{fig:closedfaceeffect} for a pictorial description.
Then clearly,  the maximum signal power between $x_i$ and $x_j$ is  $\delta^{-\alpha}$. Moreover, the interference received at either $x_i$ or $x_j$ from the nodes inside $F_{\rho}$ is greater than $\delta^{-\alpha}/T$, since $\mu \le \delta T^{\frac{1}{\alpha}}$. Thus, $\SIR_{ij} < T$ and $\SIR_{ji} < T$, and hence $x_i$ and $x_j$ cannot connect to each other. 
Similarly,  conditions 2) and 3) imply that $x_i$ or $x_j$  cannot connect to any of the nodes inside $F_{\rho}$, since $\rho \le \eta  T^{\frac{1}{\alpha}}$.  \begin{defn} A circuit in $\bH$ is a sequence of consecutive faces of $\bH$ such that the first and last face of the sequence have a common edge. A circuit in $\bH$ is defined to be open/closed if all the faces of the circuit are open/closed in $\bH$. A closed circuit is illustrated in Fig. \ref{fig:closedcircuit}.
\end{defn}

Thus, if there is a closed circuit  around the origin, then nodes of $\Phi$ lying inside the closed circuit cannot connect to any nodes of $\Phi$ outside 
the closed circuit as shown in Fig. \ref{fig:closedcircuit}. Therefore if there exists a closed circuit around the origin a.s., then a.s. there is no percolation, since infinitely many nodes of $\Phi$ cannot lie in a bounded area (inside of the closed circuit). From \cite{Grimmett1980}, we know that for a hexagonal lattice, where the probability of any face being open/closed is independent, if $P(\text {closed face}) > \frac{1}{2}$, then there exists a closed circuit around the origin a.s.. Next, we show that for large enough $T$, the probability of a face being closed is greater than $\frac{1}{2}$ when $\lambda$ lies in a  closed interval.

\begin{thm} For the PSG, $\exists \ T^{\star}$ such that for $T > T^{\star}$, $\exists$ $\Lambda_T^l=[\lambda_T^{l1} \ \lambda_T^{l2}] \subset \bbR$, such that if 
$\lambda \in \Lambda_T^l$, then the probability of percolation is  zero. 
\end{thm}
\begin{proof}
Recall from conditions 1)-4), $P(\text{closed face}) = P(\text{closed} \ F_{\delta})^6$. Note that 
\begin{eqnarray}\nn
P(\text{closed} \ F_{\delta}) &=& P(|F_{\delta} \backslash F_{\rho}|=0, |F_{\rho}|>1,  \\ \nn
&& \rho \le \eta  T^{\frac{1}{\alpha}}, \mu \le \delta T^{\frac{1}{\alpha}}), \\ \nn
&=& P(|F_{\delta} \backslash F_{\rho}|=0)P(|F_{\rho}|>1), \\ \nn
&&  \rho \le \eta  T^{\frac{1}{\alpha}}, \mu \le \delta T^{\frac{1}{\alpha}}), \\ \nn 
&=& e^{{-\lambda}\nu(F_{\delta} \backslash F_{\rho})} \\ \label{eq:closedface}
&&\left[1- e^{{-\lambda}\nu(F_{\rho})} - \lambda\nu(F_{\rho})e^{{-\lambda}\nu(F_{\rho})}\right], \\ \nn 
&&   \rho \le \eta  T^{\frac{1}{\alpha}}, \mu \le \delta T^{\frac{1}{\alpha}},
\end{eqnarray}
where $\nu(.)$ represents the Lebesgue measure on $\bbR^2$. Note that $\mu \le \delta T^{\frac{1}{\alpha}}$ is 
automatically satisfied for $T>1$, since $\mu \le \delta$ by construction. 
For large enough $T$, $\eta$ can be chosen small enough to make $\nu(F_{\delta} \backslash F_{\rho})$ small enough with $\rho \le \eta  T^{\frac{1}{\alpha}}$. From (\ref{eq:closedface}), it follows that if $\nu(F_{\delta} \backslash F_{\rho})$ is small enough, for large enough $\delta$, there exists $\Lambda_T^l = [\lambda_T^{l1} \ \lambda_T^{l2}] \subset \bbR$, where if  $\lambda \in \Lambda_T^l$, then $P(\text{closed} \ F_{\delta}) >(\frac{1}{2})^{\frac{1}{6}}$. Thus, we have shown that for a large enough $T$, there exists a closed interval, such that if $\lambda$ belongs to the closed interval then the $P(\text{closed face}) > \frac{1}{2}$, and hence  the probability of percolation is zero.
\end{proof}

{\it Discussion:} In this section we mapped the continuum percolation on the SIR graph into discrete percolation on the hexagonal lattice to make use of the known results on the discrete percolation on the hexagonal lattice. 
It is well known that if the probability of any hexagonal face being closed is more than $\frac{1}{2}$, then almost surely, the connected component of hexagonal lattice is finite. Then we showed that with our mapping, for large enough $T$, 
the probability of a closed face of the hexagonal lattice can be made more than $\frac{1}{2}$ for a closed interval of node densities, and hence almost surely the connected component of the SIR graph is finite.

\begin{figure}
\centering
\includegraphics[width=3.5in]{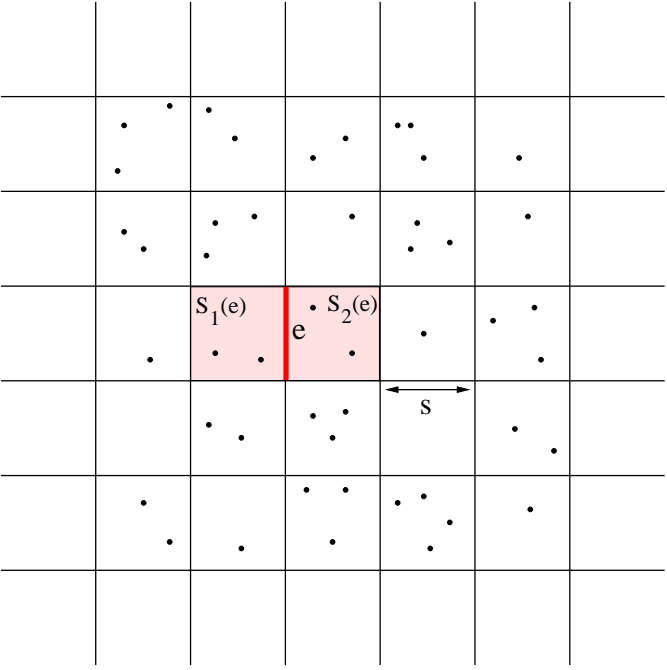}
\caption{Two-dimensional square lattice with edge $s$.}
\label{fig:squarelattice}
\end{figure}

\begin{figure}
\centering
\includegraphics[width=3.5in]{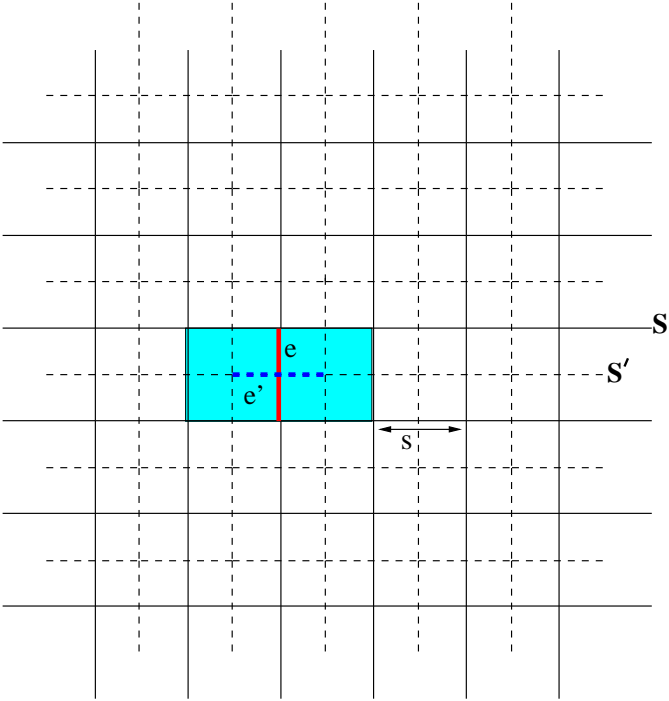}
\caption{Dual of the square lattice.}
\label{fig:duallattice}
\end{figure}

\subsection{Super-critical regime }
\label{sec:plsupc}
In this section, we show that for small enough $T$, there exists a closed interval $\Lambda_T^{u}\bydef [\lambda_T^{u1} \ \lambda_T^{u2}]$, such that  if  $\lambda \in \Lambda_T^{u}$, 
then the probability of percolation is greater than zero. For the proof provided in this section we need  that the signal attenuation function with distance $g(.)$ is monotonically decreasing and satisfies $\int_{0}^{\infty} xg(x) dx < \infty$. Clearly, $g(x) = x^{-\alpha}$ is not a valid signal attenuation function for this subsection. 

In this section we tile $\bbR^2$ into a square lattice, and define each edge to be {\it open} or {\it closed} to tie up the continuum percolation on the PSG with the percolation on the square lattice. Let $\bS$ be a square lattice with side $s\bydef \frac{1}{\sqrt{5}}g^{-1}\left(MT\right)$ as shown in Fig. \ref{fig:squarelattice}, where $M \in \bbR$ which will be chosen later. 
Let  $\bS' = \bS+(\frac{s}{2}, \frac{s}{2})$ be the dual lattice of $\bS$ obtained by translating each edge of $\bS$ by $(\frac{s}{2}, \frac{s}{2})$ as shown in Fig. \ref{fig:duallattice}. 
Any  edge $e$ of $\bS$ is defined to be open if there are one or more than one nodes of $\Phi$ in both the adjacent squares $S_1(e)$ and $S_2(e)$ as shown in Fig. \ref{fig:squarelattice}, and the interference received $I_j^i = \sum_{k\in \Phi, k\ne i }g(d_{kj})$ at any node $j\in S_1(e)\cup S_2(e)$ is less than $M, \ \forall \ i \in S_1(e)\cup S_2(e)$. Any edge of $\bS$ is defined to be closed if it is not open. Any edge $e' \in \bS'$ is defined to be open if and only if the corresponding edge $e \in \bS$ is open.
 Some important properties of $\bS$ and $\bS'$ are as follows.
\begin{defn} Open component of $\bS$ is the sequence of connected open edges of $\bS$. 
\end{defn}

\begin{lemma}\label{lem:int1} If the cardinality of the open component of $\bS$ containing the origin is infinite, then $|C_{x_1}|= \infty$. 
\end{lemma}
\begin{proof} 
Note that if an edge $e\in \bS$ is open, then all the nodes lying in $S_1(e)\cup S_2(e)$ are connected to each other, since the distance between  any two of them is less than $\frac{1}{\sqrt{5}}g^{-1}\left(MT\right)$, and hence the signal power is greater than $MT$, while the interference power is less than $M$, implying that $\SIR_{ij} > T, \ x_i, x_j \in S_1(e)\cup S_2(e)$. Thus, if there are infinite number of connected open edges in $\bS$, then the number of connected nodes of $\Phi$ is also infinite.
\end{proof}

\begin{defn} A circuit in $\bS$ or $\bS'$ is a connected path of $\bS$ or $\bS'$ which starts and ends at the same point. 
A circuit in  $\bS$ or $\bS'$ is defined to be open/closed if all the edges on the circuit are open/closed in $\bS$ or $\bS'$. 
\end{defn}
\begin{lemma}\label{lem:int2}\cite{Grimmett1980} The open component of $\bS$  containing the origin is  finite if and only if  there is a closed circuit in $\bS'$  surrounding the origin.
\end{lemma}

Next, we will show that for small enough $T$, $\exists \ \Lambda_T^u=[\lambda_T^{u1} \ \lambda_T^{u2}]$, such that if 
$\lambda \in \Lambda_T^u$, then probability of having a closed circuit in $\bS'$  surrounding the origin is less than one, and hence the probability of having an infinite open component of $\bS$ containing the origin is greater than zero. We take an approach similar to \cite{Dousse2006}.

Let $A_e=1$ if  there are one or more than one nodes of $\Phi$ in both the adjacent squares $S_1(e)$ and $S_2(e)$ of $e$, and zero otherwise. Similarly, let $B_e=1$ if   the interference $I_j^i$ received at any node $j \in S_1(e)\cup S_2(e)$ is less than $M, \ \forall \ i \in S_1(e)\cup S_2(e)$ and zero otherwise. 
Then by definition, an edge $e \in \bS$ is open if $C_e = A_e B_e=1$. Now we want to bound the probability of having a closed circuit surrounding the origin in $\bS$. Towards that end, we will first bound the probability of a closed circuit of length $n$, i.e. $P(C_1=0, C_2=0, \ldots, C_n=0), \ \forall \ n \in \bbN$ considering $n$ distinct edges. 
Let $p_A \bydef P(A_n=0)$ for any $n$. Since $\Phi$ is a PPP with density $\lambda$,  $p_A = 1-(1-e^{-\lambda s })^2$. Then we have the following intermediate results to upper bound $P(C_1=0, C_2=0, \ldots, C_n=0)$.
\begin{lemma}\label{lem:a} $P(A_1=0, A_2=0, \ldots, A_n=0) \le   p_1^{n}$, where $p_1 = p_{A}^{1/4}$. 
\end{lemma}
\begin{proof} Follows from the fact that in any sequence of  $n$ edges of $\bS $ there are at least $n/4$ edges such that their adjacent squares 
$S_1(e)\cup S_2(e)$ do not overlap. Therefore $P(A_1=0, A_2=0, \ldots, A_n=0) \le P(\cap_{e\in O}A_e=0)$, where $O$ is the set of edges for which their adjacent squares  $S_1(e)\cup S_2(e)$ have no overlap, and  $|O|=n/4$. Since $S_1(e)\cup S_2(e), \ e \in O$ have no overlap, and events $A_e=0$ are independent for $e\in O$,  the result follows. 
\end{proof}

\begin{lemma}\label{lem:b} \cite[Proposition 2]{Dousse2006} For $\int_{0}^{\infty} xg(x) dx < \infty$, $P(B_1=0, B_2=0, \ldots, B_n=0) \le  p_2^{n}$, where $p_2 \bydef e^{\left(\frac{2\lambda}{K}\int g(x) dx - \frac{M}{K} \right)}$, and $K$ is a constant. 
\end{lemma}

\begin{lemma}\label{lem:c} \cite[Proposition 3]{Dousse2006} $P(C_1=0, C_2=0, \ldots, C_n=0) \le (\sqrt{p_1} + \sqrt{p_2})^n$. 
\end{lemma}

Let $q \bydef (\sqrt{p_1} + \sqrt{p_2})$. Using the Peierl's argument, the next Lemma characterizes an upper bound on $q$ for which having a closed circuit in $\bS$  surrounding the origin is less than one.
\begin{lemma}\label{lem:suffcond} If $q < \frac{11-2\sqrt{10}}{27}$, then the probability of having a closed circuit in $\bS'$  surrounding the origin is less than one.
\end{lemma}
\begin{proof} From \cite{Grimmett1980}, the number of possible circuits of length $n$ around the origin is less than or equal to  $4n3^{n-2}$. From Lemma \ref{lem:c}, we know that the probability of a closed circuit of length $n$ is upper bounded by $q^n$. Thus, 
\begin{eqnarray*}
P(\text{closed circuit around origin}) &\le & \sum_{n=1}^{\infty} 4n3^{n-2} q^n,\\
&=& \frac{4q}{3(1-3q)^2},
\end{eqnarray*}
which is less than $1$ for $q < \frac{11-2\sqrt{10}}{27}$.
\end{proof}

\begin{thm} For the PSG,  where the attenuation function $g(.)$ is monotonically decreasing and satisfies $\int xg(x)dx < \infty$, for small enough $T$, $\exists \ \Lambda_T^u=[\lambda_T^{u1} \ \lambda_T^{u2}]$, such that if 
$\lambda \in \Lambda_T^u$, then the probability of percolation on the PSG is greater than zero. 
\end{thm}
\begin{proof} From Lemma  \ref{lem:suffcond}, we know that if $q < \frac{11-2\sqrt{10}}{27}$, then the probability of having a closed circuit in $\bS'$ is less than $1$. Hence from Lemma \ref{lem:int2},  if $q < \frac{11-2\sqrt{10}}{27}$, then the probability of percolation on the PSG is greater than zero. Recall  that 
$q = \sqrt{p_1} + \sqrt{p_2}$, where  $p_1 = (1-(1-e^{-\lambda s })^2)^{1/4}$, $s=\frac{1}{\sqrt{5}}g^{-1}\left(MT\right)$, and $p_2 \bydef e^{\left(\frac{2\lambda}{K}\int g(x) dx - \frac{M}{K} \right)}$. Next, we show that $q$ can be made arbitrarily small for a closed interval $\Lambda_T^u=[\lambda_T^{u1} \ \lambda_T^{u2}]$ by appropriately choosing $M$ for small enough $T$. 
Let $M=1/T$, then $p_1$ does not depend on $M$ or $T$,  and $p_1$  decreases to zero with increasing $\lambda$. Moreover, for small enough $T$ with 
$M=1/T$, depending on $K$, $p_2$ can be made arbitrarily small for values of $\lambda$ for which $p_1$ is very small. Thus, for small enough $T$, there exists a value of $\lambda$ for which $q < \frac{11-2\sqrt{10}}{27}$. Moreover, since $q$ is a continuous function, there exists a closed interval $\Lambda_T^u=[\lambda_T^{u1} \ \lambda_T^{u2}]$ for which $q < \frac{11-2\sqrt{10}}{27}$, and consequently for $\lambda \in \Lambda_T^u$, the probability of percolation on the PSG is greater than zero.
\end{proof}

{\it Discussion:} In this section we mapped the continuum percolation on the SIR graph into discrete percolation on the square lattice. With a square lattice, it is known that if the probability of having a closed circuit around the origin is less than one, then with positive probability an unbounded connected cluster is present in the square lattice. 
Then with our mapping, for small enough $T$, we showed that 
the probability of having a closed circuit around the origin is less than one for a closed interval of node densities. 
Consequently, for small enough $T$, we  concluded that the connected cluster of the SIR graph is unbounded for a closed interval of node densities. Since percolation happens for all non-zero values of $\lambda$ at $T=0$, by showing that percolation happens for small enough $T$, our result loosely establishes the continuity of percolation at $T=0$.

Even though our result is only valid for small enough $T$, we expect that for any value of $T$, percolation can happen only for a "small" closed interval of node densities, if at all. The justification for this claim is that  for extremely small values of node densities, the minimum distance between nodes is large, and it is unlikely that SIR for large number of nodes is larger than $T$, while for extremely large values of node densities, interference is significant and it is difficult for sufficient number of nodes to have SIR greater than $T$.

After having established that percolation happens on the SIR graph for a small enough threshold $T$, the next natural question to ask is: whether the SIR graph is connected for small enough $T$, where by a connected graph we mean that there is a path for each node to any other node in the graph. Since the probability of SIR graph being connected in an infinite plane with any node density is zero, we restrict ourselves to 
an unit square where $n$ nodes are uniformly distributed, and ask the question whether the SIR graph restricted to an unit square is connected for small enough $T$ in the next Section.

\section{Connectivity on the SIR graph}
 For studying the SIR graph connectivity, we restrict ourselves to an 
unit square and assume that $n$ nodes of $\Phi_n = \{x_1, \dots, x_n\}$  are drawn independently from an uniform 
distribution on the unit square. Following Section \ref{sec:sys}, the SIR graph on the unit square is defined as $SG(T,1) \bydef \{\Phi_n, {\cal E}_n\}$, where ${\cal E}_n = \{(x_i,x_j) \ : \SIR_{ij} \ge T\}$, and 
$\SIR_{ij} \bydef \frac{g(d_{ij})} 
{\sum_{k\in \Phi_n, k\ne i }g(d_{kj})}$.
\begin{defn} The SIR graph $SG(T,1)$ is defined to be connected if there is a path from $x_i \rightarrow x_j \ \text{in} \ SG(T,1), \ \forall i,j = 1,2,\dots,n, i\ne j$.
\end{defn}

\begin{figure}
\centering
\includegraphics[width=3.5in]{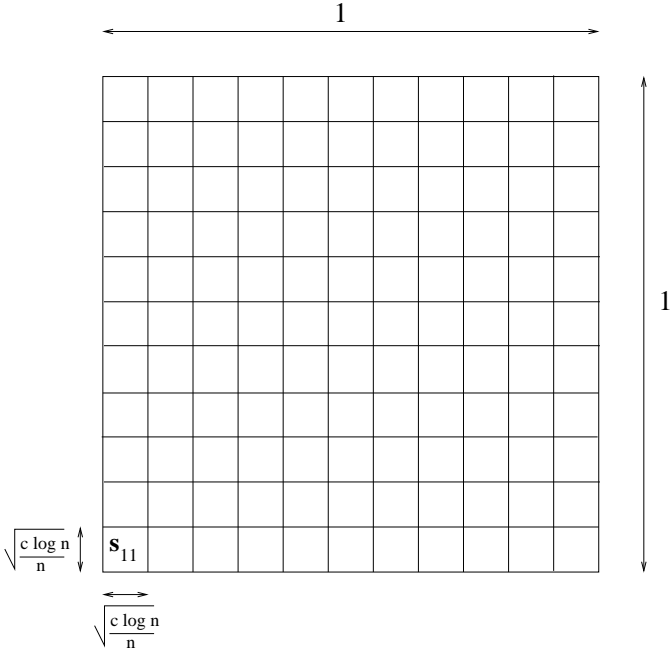}
\caption{Square tiling of the unit square.}
\label{fig:sqlattice}
\end{figure}

To analyze the connectivity of the SIR graph, we color the nodes of $\Phi$ with $C(n)$ colors, where nodes assigned different colors 
correspond to having orthogonal signals in either time or frequency. 
Graph coloring is a map ${\cal C}: \Phi \rightarrow [C(n)]$, such that ${\cal C}(x_m) = c(x_m), c(x_m)\in [C(n)]$. 
Coloring ensures that only those signals transmitted from the similarly colored nodes interfere with each other.
Then the colored SIR graph  is 
defined as 
$SG(T,1,C(n)) \bydef \{\Phi_n, {\cal E}_n\}$, where ${\cal E}_n = \{(x_i,x_j) \ : \SIR_{ij} \ge T\}$, and 
$\SIR_{ij} \bydef \frac{g(d_{ij})} 
{\sum_{k\in \Phi_n, k\ne i, c(x_k)=c(x_i) }g(d_{kj})}$, and $SG(T,1,C(n))$ is defined to be connected if there is a path from 
$x_i \rightarrow x_j \ \text {in} \ SG(T,1,C(n)) \ \forall, i,j = 1,2,\dots,n, i\ne j$. 

Note that  $SG(T,1) = SG(T,1,1)$. In the next Theorem, we find  an upper bound on $C(n)$ for which $SG(T,1,C(n))$ is connected. 


\subsection {Upper bound on $C(n)$}
For generality, we will prove the upper bound for the singular path-loss model $g(d_{ij}) = d_{ij}^{-\alpha}$, which easily extends to all other path-loss models with monotonically decreasing $g(.)$ and $\int xg(x) dx < \infty$. 
The main result of this subsection is as follows. 
\begin{thm}\label{thm:coloring} If $C(n) > 4 (1+\delta) c \log n$ colors are used for coloring the  SIR graph $SG(T,1, C(n))$, where $c$ and $\delta$ are independent of $n$, then the SIR graph $SG(T,1,C(n))$ is connected with high probability.  
\end{thm}
\begin{proof}
Consider a $1\times 1$ square $\bS_1$.  We assume that $n$ nodes are distributed uniformly in $\bS_1$. We tile $\bS_1$ into 
smaller squares $\bs_{ij}$ with side $\sqrt{\frac{c \log n}{n}}$ as shown in Fig. \ref{fig:sqlattice}. Let the number of nodes lying in $\bs_{ij}$ be $|\bs_{ij}|$.  
Let the set of colors to be used be $C(n) \bydef \{c_1, c_2, c_3, c_4\}$, where $|c_{\ell}| = (1+\delta) c \log n, \ \ell=1,2,3,4$, and 
$c_{\ell} \cap c_{k} = \phi, \forall \ \ell, k$. Colors from set $c_1$ and $c_2$ are associated with alternate rows in odd numbered columns, while sets $s_3$ and $s_4$ are associated with alternate rows in even numbered columns  in the tilting  of $\bS_1$ using $\bs_{ij}$ as shown in Fig. \ref{fig:colorsqlattice}. Nodes in each smaller square $\bs_{ij}$ are colored as follows.
Let the nodes lying in each $\bs_{ij}$ be indexed using 
numbers  $1$ to $|\bs_{ij}|$. Then we associate $(1+\delta) c \log n$ colors to each $\bs_{ij}$ in a regular fashion, i.e. color of node 
$p, \ p=1,\dots, |\bs_{ij}|$ is  $ p \mod (1+\delta) c \log n$. Since $\bbE\{|\bs_{ij}|\} = c \log n, \ \forall \ i, j$, from the Chernoff bound,  $P\left(|\bs_{ij}| > (1+\delta) c \log n\right) \le n^{\frac{-c\delta^2}{3}}$. Hence with this coloring, the probability that there are two or more nodes with the same color in a given square $\bs_{ij}$ is
\begin{equation}\label{chernoffatmax}
P(\text{two nodes with the same color in}\ \bs_{ij}) \le n^{\frac{-c\delta^2}{3}}.
\end{equation}

%

\begin{figure}
\centering
\includegraphics[width=2.5in]{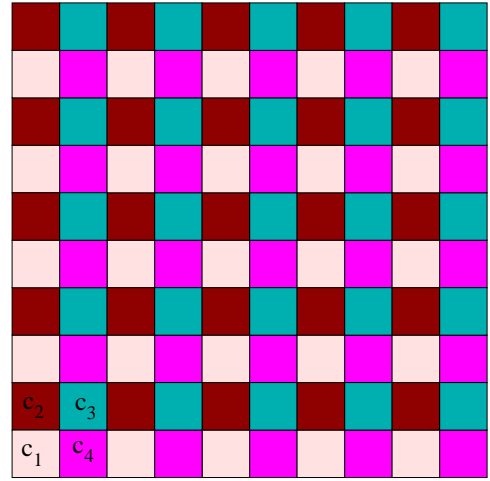}
\caption{Coloring the square tiling of the unit square with four sets of colors.}
\label{fig:colorsqlattice}
\end{figure}

Consider another square $\bs_{t}(m)$ with side $\sqrt{\frac{m \log n}{n}}$ centered at any node $x_t$  as shown in Fig. \ref{fig:intfsqlattice}, where $m<c$ is a constant. 
Again, using the Chernoff bound as above, we  have that \begin{equation}\label{chernoffatleast} P\left(|\bs_{t}(m)| < \frac{m}{2} \log n\right) \le n^{-2}.\end{equation}

Now define events $E_{ij} = \{\text{two nodes with the same color in} \  \bs_{ij}\}$, and \[F_{t}(m) = \{\text{there are less than} \ \frac{m}{2} \log n \  \text{nodes in}  \ \bs_{t}(m)\}.\] Then the probability that the SIR graph is connected $P(SG(T,1,C(n))\ \text{is connected})$ can be written as $P(SG(T,1,C(n))\ \text{is connected}) $
\begin{eqnarray*} \nn
&=& P(\cap_{ij}E_{ij} \cup \cap_{t} F_{t}(m))\\
&&P\left(SG(T,1,C(n))\ \text{is connected} \ | \cap_{ij}E_{ij} \cup \cap_{t} F_{t}(m) \right) \\ \nn
&&+\  P(\left(\cap_{ij}E_{ij} \cup \cap_{t} F_{t}(m)\right)^c)\\
&&P\left(SG(T,1,C(n))\ \text{is connected} \ | \left( \cap_{ij}E_{ij} \cup \cap_{t} F_{t}(m)\right)^c \right).
\end{eqnarray*}

Using the union bound over all squares $\bs_{ij}$, and over all nodes $x_t$, together with (\ref{chernoffatmax}), and (\ref{chernoffatleast}), it follows that $ P(\cap_{ij}E_{ij} \cup \cap_{t} F_{t}(m)) \le n^{-1} + n^{1-\frac{c \delta^2}{3}}$. Thus, for large enough $n$, 
\[P(SG(T,1,C(n))\ \text{is connected})\sim \]\[\ \ \ \ \  \ P\left(SG(T,1,C(n))\ \text{is connected} \ | \left(\cap_{ij}E_{ij} \cup \cap_{t} F_{t}(m)\right)^c \right).\] 
Hence in the sequel, we analyze the SIR connectivity while conditioning on the event that no square $\bs_{ij}$ has more than two nodes with the same color, and each square $\bs_{t}(m)$ has at least $\frac{m}{2} \log n$ nodes.

\begin{figure*}
\centering
\includegraphics[width=4.5in]{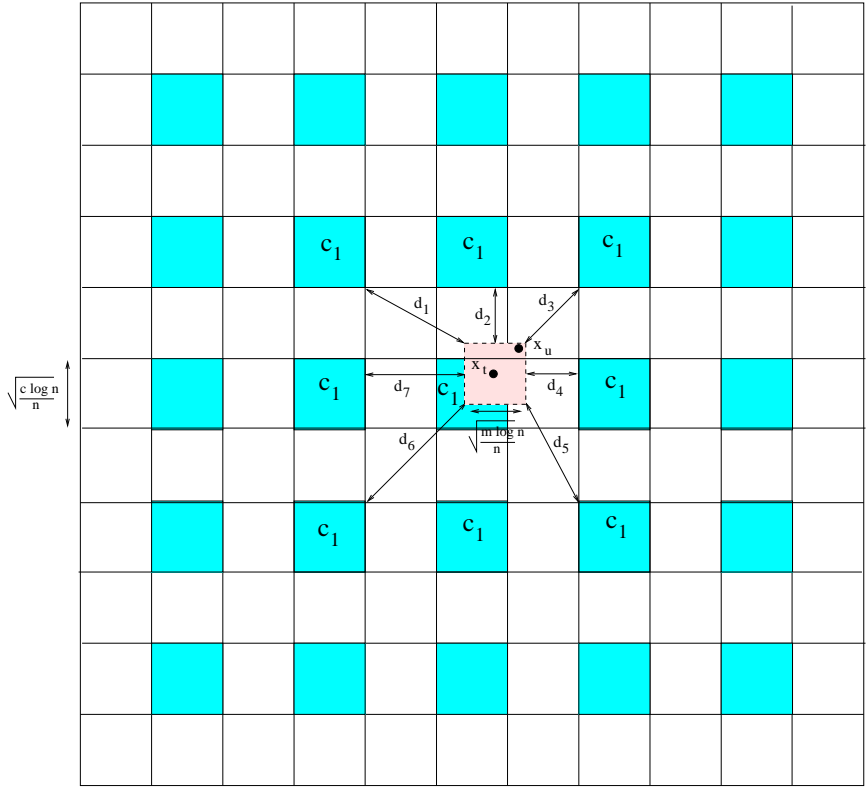}
\caption{Pictorial description of distance from nearest interferers after coloring.}
\label{fig:intfsqlattice}
\end{figure*}
Now, under the conditioning, to show that the SIR graph is connected, it is sufficient to show that for any $t=1,\dots, n$, $x_t$ is connected to all nodes in $\bs_{t}(m)$ in the SIR graph. 
Towards that end, let $x_u, u\ne t$ be any other node in $\bs_{t}(m)$. Then the distance between $x_t$ and $x_u$, $d_{tu}$, is upper bounded by 
$ \sqrt{\frac{2m \log n}{n}}$, thus the signal power $d_{tu}^{-\alpha} \ge \left(\frac{2m \log n}{n}\right)^{-\alpha/2}$. 
Now consider Fig. \ref{fig:intfsqlattice} for analyzing the interference power at $x_u$. Without loss of generality assume that $x_t$ belongs to the square associated with color set $c_1$.
Note that even if $x_t$ and $x_u$ belong to the same square $\bs_{ij}$, there is no other node in $\bs_{ij}$ that has the same color as $x_t$. So the interference received at $x_t$ is attributed to nodes lying in square $\bs_{i'j'}$ associated with color set $s_1$, where either 
$i'\ne i$ or $j'\ne j$. From Fig. \ref{fig:intfsqlattice}, it is clear that for any $q, \ q=1,2,\dots,n,$ there are maximum $8$ nodes using the same color as $x_t$, at a distance at least $2q \left(\sqrt{\frac{c \log n}{n}} - \sqrt{\frac{m \log n}{n}}\right)$ from $x_u$, since $d_1, \dots, d_7 \ge \left(\sqrt{\frac{c \log n}{n}} - \sqrt{\frac{m \log n}{n}}\right)$. Thus, the  interference power 
\[ \sum_{v\ne t, c(x_v) = c(x_t)} d_{vu}^{-\alpha} \le \sum_{q=1}^n \frac{8}{\left(2q \left(\sqrt{\frac{c \log n}{n}} - \sqrt{\frac{m \log n}{n}}\right)\right)^\alpha}. \]
Since $m<c$ is a constant, \[\left(\sqrt{\frac{c \log n}{n}} - \sqrt{\frac{m \log n}{n}}\right) \ge \beta \left(\sqrt{\frac{c \log n}{n}}\right)\] for some $\beta >0$. Hence 
\begin{equation} 
\sum_{v\ne t, c(x_v) = c(x_t)} d_{vu}^{-\alpha} \le \frac{8}{(2 \beta)^{\alpha} \left(\frac{c \log n}{n}\right)^{\alpha/2}} \sum_{q=1}^n  q^{-\alpha}.
\end{equation}
Since $\sum_{q=1}^n  q^{-\alpha}$ converges for $\alpha \ge 2$, let  $c_5 \bydef \sum_{q=1}^n  q^{-\alpha}$. Then, computing the 
SIR, we have 
\[\SIR_{tu} = \frac{d_{tu}^{-\alpha}}{\sum_{v\ne t, c(x_v) = c(x_t)} d_{vu}^{-\alpha}} \ge \frac{ c^{\alpha/2} \left(\frac{m}{2\beta}\right)^{-\alpha/2}   } {8c_5},\] which can be made more than $T$, the SIR threshold, by appropriately choosing $c$ and $m$. For example, for $\alpha =2$, for $c > \frac{3\beta T}{2\pi^2m}$, $SIR_{tu} > T$.
Thus,  for an appropriate choice of $c$ and $m$, $P\left(SG(T,1,C(n))\ \text{is connected} \ | \left(\cap_{ij}E_{ij} \cup \cap_{t} F_{t}(m)\right)^c \right) =1$, and  \[\lim_{n\rightarrow \infty }P(SG(T,1,C(n))\ \text{is connected}) =1.\]
\end{proof}

{\it Discussion:} Theorem \ref{thm:coloring} implies that ${\cal O}(\log n)$ colors are sufficient for guaranteeing the connectivity of $SG(T,1,C(n))$ with high probability.  The intuition behind this result is that if only $n/log(n)$ nodes interfere with any node's transmission then the total interference received at any node is bounded with high probability, and each node can connect to a large number of nodes. In the next  subsection we show that actually $C(n) = {\cal O}(\log n)$ colors are also necessary for the $SG(T,1,C(n))$ to be connected with high probability, and if $C(n)$ is less than order $\log n$, then the interference power can be arbitrarily large and difficult to bound, making  
$SG(T,1)$ disconnected with high probability. 

\begin{rem} Recall that SIR connectivity has been studied in \cite{Gupta2000} under the physical model, where it is shown that if simultaneously transmitting nodes are at least $\Delta$ distance away, then all the nodes within a fixed radius from the active transmitters have SIR's greater than the specified threshold for large enough $\Delta$. The result of \cite{Gupta2000}, however, is valid only for $\alpha >2$. In comparison, our result is valid for all $\alpha$ for which $\sum_{n=1}^{\infty} n^{-\alpha}$ is finite. Our approach is similar to SIR connectivity analysis of the one-dimensional case \cite{Avin2010}, where $n$ nodes are uniformly distributed in the unit interval. 
\end{rem}

\subsection {Lower bound on $C(n)$}\label{sec:lbconn}
In this section we show that if less than order $\log(n)$ colors are used, then the SIR graph is disconnected with high 
probability. To show this, we actually show that any node is not connected to any other node with high probability if 
less than order $\log(n)$ colors are used. For proving this lower bound we will restrict ourselves to path-loss models with monotonically decreasing $g(.)$ and $\int xg(x) dx < \infty$, since with singular path-loss models, $g(d_{ij}) = d_{ij}^{-\alpha}$, the signal power between any two nodes cannot be bounded.
Formally, our result is as follows.

 \begin{thm} For path-loss models with monotonically decreasing $g(.)$ and $\int xg(x) dx < \infty$, if $C(n) = \frac{Tf(n)}{\omega}$, where $\lim_{n\rightarrow \infty} \frac{f(n)}{\log n} = 0$, i.e. $C(n)$ is sub-logarithmic in $n$, and $\omega$ is a constant,  then the SIR graph $SG(T,1,C(n))$ is not connected with high probability.  
\end{thm}

\begin{proof}
To show that $C(n) = {\cal O}( \log n)$ is necessary for guaranteeing the connectivity of $SG(T,1,C(n))$ with high probability, similar to last subsection, we consider the tiling of  the unit square $\bS_1$ by squares $\bs_{ij}$, but with  side $\sqrt{\frac{\log n}{n}}$, instead of $\sqrt{\frac{c\log n}{n}}$ as shown in Fig. \ref{fig:squarelattice}.

With this tiling, the expected number of nodes in any square $\bbE\{|\bs_{ij}|\} = \log n$, and $P(|\bs_{ij}| < (1-\delta) \log n) \le n^{-\delta^2/2}$, for any $0<\delta<1$. Therefore with $C(n) = \frac{T f(n)}{\omega}$ colors, where $\lim_{n\rightarrow \infty} \frac{f(n)}{\log n} = 0$, with high probability, there are at least $\omega/T$ nodes in each square using one particular color $c_p\in C(n)$. Let $\Phi_{c_p} = \{x_m :  c(x_m) = c_p, x_m \in \bs_{ij}\}$ be the set of nodes in square $\bs_{ij}$ that use the  color $c_p$.  Note that $|\Phi_{c_p}| > \omega/T$ with high probability. Consider two nodes $x_k, x_{m} \in \Phi_{c_p}$, and any other node $ x_{\ell} \in \bs_{ij}$. By the definition of $\bs_{ij}$,
the distance between node $x_m$ and $x_{\ell}$, $d_{m\ell}$ is no more than $d_{k\ell} + \sqrt{\frac{2\log n}{n}}$ . Therefore the interference received at $x_{\ell}$ from nodes inside $\bs_{ij}$ using color $c$ is 
$\sum_{x_m\in \Phi_{c_p}, m\ne k} g\left(d_{m\ell}\right)$ which is greater than $ (\omega/T-1) g\left(d_{k\ell} + \sqrt{\frac{2\log n}{n}}\right)$ since $|\Phi_{c_p}| > \omega/T$. Thus the SIR between $x_k$ and $x_{\ell}$ is 
\begin{eqnarray*}
\SIR_{k\ell} &\le & \frac{g(d_{k\ell})}{(\omega/T-1) g\left(d_{k\ell} + \sqrt{\frac{2\log n}{n}}\right)}.
\end{eqnarray*}
Since $g(.)$ is bounded, choosing $\omega$ appropriately, $\SIR_{k\ell} < T$. Thus, we have shown that node $x_k$ is not connected to any node inside $\bs_{ij}$. Similarly, it follows that $x_k \in \bs_{ij}$ is not connected to any node outside of $\bs_{ij}$, since for $x_p \notin \bs_{ij}$, the signal power $g(d_{kp})$ is less 
compared to $g(d_{k\ell})$ the signal power at any node $x_{\ell} \in \bs_{ij}$, while the interference powers at $x_{p} \notin \bs_{ij}$ and $x_{\ell} \in \bs_{ij}$ are identical. Thus, we conclude that if less than order $\log n$ colors are used, then $SG(T,1,C(n))$ is not connected with high probability.
\end{proof}

{\it Discussion:} In this subsection we  showed that if less than order $\log n$ colors are used, then the SIR graph $SG(T,1,C(n))$ is disconnected with high probability. This result holds for any SIR threshold $T$, and hence even for small enough $T$, the  SIR graph cannot be connected by using a single color. This result is in contrast to our percolation result where we showed that for small enough $T$, percolation happens for a closed interval of node densities.


\section{Conclusion}
In this paper we studied the percolation and connectivity properties of the SIR graph. The analysis is complicated 
since the link formation between any two nodes depends on all the other nodes in the network (through their 
interference contribution) and entails infinite range dependencies. For studying percolation on the SIR graph, we 
tied up the continuum percolation on the SIR graph to discrete percolation for which prior results are known. For 
finding a sub-critical regime, we made use of the hexagonal lattice, while for the super-critical regime percolation on 
the square lattice is considered. We showed the existence of a closed interval of node 
intensities for which the SIR graph percolates or not depending on the SIR threshold. 
Ensuring connectivity is a stricter condition compared to percolation, since with connectivity every pair of nodes should have a path between them. We took the graph coloring approach for studying connectivity on the SIR graph, and found 
upper and lower bounds on the number of colors required for guaranteeing connectivity with high probability. The derived upper and lower bounds are tight, and from which we conclude that using colors that are logarithmic in the number of nodes is necessary and sufficient for ensuring connectivity in the SIR graph with high probability.

\bibliographystyle{../../IEEEtran}
\bibliography{../../IEEEabrv,../../Research}

\end{document}